\documentclass[10pt,a4paper]{article}
\makeatother

\usepackage{amsmath}
\usepackage[a4paper, total={6in, 8in}]{geometry}
\usepackage{amsthm}
\usepackage{verbatim}
\usepackage{graphicx}
\usepackage{esvect}
\usepackage{cite}
\usepackage{amssymb}
\usepackage{footnote}
\usepackage{color}
\usepackage{url}
\usepackage{subcaption}
\usepackage{algpseudocode}
\newtheorem{theorem}{Theorem}[section]

\newtheorem{definition}{Definition}[section]
\newtheorem{lemma}[theorem]{Lemma}
\newtheorem{proposition}{Proposition}[section]
\newtheorem{remark}{Remark}
\newtheorem*{problem statement}{Problem Statement}

\newtheorem{assumption}{Assumption}
\newtheorem{feasibility problem}{Problem}
\newcommand{\setz}{\mathrm{Z}}
\newcommand{\sete}{\mathrm{E}}
\graphicspath{{./Figures/}}

\newcommand{\real}{\mathbb{R}}
\newcommand{\sys}{\mathcal{S}}
\newcommand{\reg}{\mathcal{R}}
\newcommand{\reach}{\mathcal{X}}
\newcommand{\cone}{\mathcal{C}}
\newcommand{\sphere}{\mathcal{S}}
\begin{document}
	\title{Traffic Abstractions of Nonlinear Homogeneous Event-Triggered Control Systems}
	\author{Giannis Delimpaltadakis and Manuel Mazo Jr.\thanks{The authors are with the Delft Center for Systems and Control, Delft university of Technology, Delft 2628CD, The Netherlands. Emails:\texttt{\{i.delimpaltadakis, m.mazo\}@tudelft.nl}. This work is supported by the ERC Starting Grant SENTIENT (755953).}}
	\date{}
	\maketitle
	
\begin{abstract}
	In previous work, linear time-invariant event-triggered control (ETC) systems were abstracted to finite-state systems that capture the original systems' sampling behaviour. It was shown that these abstractions can be employed for scheduling of communication traffic in networks of ETC loops.
	In this paper, we extend this framework to the class of nonlinear homogeneous systems, however adopting a different approach in a number of steps. Finally, we discuss how the proposed methodology could be extended to general nonlinear systems.
\end{abstract}
\section{Introduction}
	The present-day ubiquity of networked control systems (NCS) has raised the research community's awareness regarding the consumption of communication bandwidth of digital control implementations. Specifically, periodic sampling seems to be inefficient, as it leads to unnecessary communication between controllers and sensors. In this context, and promising to reduce the bandwidth used by networked control loops, aperiodic schemes have been proposed: Event-Triggered Control (ETC) \cite{arzen1999, tabuada2007etc, girard2015dynamicetc} and Self-Triggered Control (STC) \cite{velasco2003self, tosample, fiter2012state, delimpa}. For an introduction to ETC/STC see \cite{2012introtoetc_stc}.
	
	Both of them are sample and hold implementations, in which at every sampling time instant the sensors transmit measurements to the controller and, only then, the controller updates the control action. These schemes exploit the system's dynamics to decide when to close the sampling loop, while guaranteeing that certain performance criteria (e.g. stability) are met. In ETC, intelligent sensors monitor the state of the plant, and transmit measurements when a certain state-dependent \textit{triggering condition} is met. On the other hand, in STC the controller is the one to decide about the sampling time, based on previous measurements. Although STC relaxes the need of an intelligent sensory system compared to ETC, it is considered less robust, due to its sampling's open loop nature, since the controller does not receive any, maybe critical, information between samples.
	
	Even though ETC has enjoyed a big share of research, there are still unresolved issues that forbid it being a widespread paradigm. According to the authors' opinion, one of the most prominent problems is the scheduling of communication traffic of ETC loops in shared networks. It is the ETC traffic's inherently aperiodic and unpredictable nature that constitutes the problem challenging. To the authors' knowledge, almost all of the approaches that are proposed to solve the problem belong to the family of controller/scheduler co-design \cite{buttazzo1998elastic, caccamo2000elastic, bhattacharya2004anytime, fontanelli2008scheduling, areqi2015scheduling, lu2002feedback, cervin2005control}. Basically, the control law, the sampling-scheme, and the scheduling of communication are all co-designed, such that resource utilization is efficient, while certain performance guarantees are met. The main drawback of these approaches is their lack of versatility, which is a result of the coupling of the controller, sampling, and scheduler design; e.g. whenever a new control loop joins the network, these techniques have to be applied again from scratch, resulting in a different design.
	
	In \cite{arman_formal_etc}, a different approach was developed for scheduling of traffic produced by linear time-invariant (LTI) ETC systems, which decouples the controller, sampling and scheduler designs, thus being more versatile. In particular, the infinite-state ETC system is abstracted by a finite-state \textit{quotient system} (or \textit{abstraction}) that captures all possible sequences of the ETC system's sampling times. In fact, it is proven that the constructed abstraction \textit{$\epsilon$-approximately simulates} (see \cite{tabuada_book}) the ETC system and can be employed for scheduling, as showcased in \cite{dieky2015symbolic}. To derive the quotient system's states, the state-space is partitioned into a finite number of cones. Afterwards, a convex embedding approach is followed, that derives lower and upper bounds of inter-event times for each conic region/state of the quotient system, which serve as outputs of the abstraction. Finally, the transitions between the abstraction's states are obtained via reachability analysis (e.g. see \cite{dreach}) conducted on each of the conic regions.
	
	In this work, the framework of \cite{arman_formal_etc} is extended to the class of nonlinear homogeneous systems. Nonetheless, the present approach is essentially different in many steps. First, a finite set of times $\{\underline{\tau}_{1},\dots,\underline{\tau}_{q}\}$ that will serve as lower bounds on inter-event times is fixed \textit{a priori}. Then, the state-space is partitioned into regions $\reg_{i,j}$, delimited by intersections of cones with \textit{inner-approximations of isochronous manifolds} that correspond to the chosen times $\underline{\tau}_i$, which were derived in \cite{delimpa}. In this way, the dynamics of the system dictate the state-space partitioning and more control on the abstraction's precision is gained. The upper bounds on inter-event times and the abstraction's transitions are determined \textit{concurrently}, via reachability analysis. To carry out the reachability analysis, an algorithm is proposed that overapproximates the transcendental sets $\reg_{i,j}$ by semi-algebraic ball-segments $\hat{\reg}_{i,j}$. Finally, in Section \ref{discussion section} it is briefly discussed how the presented methodology could be extended to general nonlinear systems. In these terms, the present work contributes to the solution of the scheduling problem of networks of ETC loops.

\section{Notation and Preliminaries}
\subsection{Notation}
	The Euclidean norm of a point $x\in\real^n$ is denoted by $|x|$. We use $\exists!$ to denote existence and uniqueness. $\real_0^+$ denotes the set of non-negative reals. Given a set $X$, $2^X$ denotes the power set of $X$. If $Q\subseteq X\times X$ is an equivalence relation on $X$, the set of all equivalence classes is denoted by $X / Q$.
	
	Consider a system of first order differential equations:
	\begin{equation}\label{ode}
		\dot{\zeta}(t) = f(\zeta(t)),
	\end{equation}
	where $\zeta:\real\to\real^n$ and $f:\real^n\to\real^n$. The solution of the above system with initial condition $\zeta_0$ is denoted as $\zeta(t;\zeta_0)$ (for simplicity we always assume that the initial time $t_0=0$). When $\zeta_0$ is clear from the context, we might omit it. The Lie derivative of a function $h$ at a point $x$ along the flow of $f$ is denoted as $\mathcal{L}_fh(x)$. Similarly, $\mathcal{L}_f^k h(x)= \mathcal{L}_f(\mathcal{L}_f^{k-1} h(x))$ is the $k$-th Lie derivative with $\mathcal{L}_f^0 h(x) = h(x)$.
\subsection{Systems and Simulation Relations}
	To introduce notions related to systems and relations between them, first we give some preliminary definitions.
	\begin{definition}[Metric \cite{combinatorial}]
		Given a set $X$, a function $d:X\times X\to\real_0^+\cup\{+\infty\}$ is a metric on $X$, if it satisfies the following properties for all $x_a,x_b,x_c\in X$:
		\begin{equation*}
			\begin{aligned}
			&\bullet \quad d(x_a,x_b) = d(x_b,x_a),\\
			&\bullet \quad d(x_a,x_b) = 0 \iff x_a=x_b,\\
			&\bullet \quad d(x_a,x_b) \leq d(x_a,x_c)+d(x_b,x_c).
			\end{aligned}
		\end{equation*}
		
	\end{definition}
	The ordered pair $(X,d)$ then forms a \textit{metric space}.
	\begin{definition}[Hausdorff Distance \cite{combinatorial}]
		Consider a metric space $(X,d)$ and two subsets $X_a,X_b\subseteq X$. The Hausdorff distance between $X_a$ and $X_b$ is defined as:
		\begin{equation*}
			d_H(X_a,X_b) := \max\Big\{\sup\limits_{x_a\in X_a}\inf\limits_{x_b\in X_b}d(x_a,x_b), \sup\limits_{x_b\in X_b}\inf\limits_{x_a\in X_a}d(x_a,x_b)\Big\}.
		\end{equation*} 
	\end{definition}
	
	We are ready to proceed to notions related to systems and relations, within the framework of \cite{tabuada_book}.
	\begin{definition}[System \cite{tabuada_book}]
		A system $\mathcal{S}$ is a tuple $(X, X_0, U,$ $\longrightarrow, Y, H)$, where $X$ is the set of states, $X_0$ is the set of initial states, $U$ is the set of inputs, $\longrightarrow \subseteq X\times U\times X$ is a transition relation, $Y$ is the set of outputs and $H:X\to Y$ is the output map.
	\end{definition}
	If $X$ is a finite (infinite) set, then $\mathcal{S}$ is called finite-state (infinite-state). A system $\mathcal{S}$ is called a \textit{metric system} if $Y$ is equipped with a metric $d:Y\times Y\to\real_0^+\cup\{+\infty\}$.
	\begin{definition}[$\epsilon$-Approximate Simulation Relation \cite{tabuada_book}]
		Consider two metric systems $\sys_a,\sys_b$ with $Y_a=Y_b$ and a constant $\epsilon\geq0$. An equivalence relation $Q\subseteq X_a\times X_b$ is an $\epsilon$-approximate simulation relation from $\sys_a$ to $\sys_b$ if it satisfies:
		\begin{itemize}
			\item $\forall x_{0_a}\in X_{0_a}: \quad \exists x_{0_b}\in X_{0_b}$ such that $(x_{0_a},x_{0_b})\in Q$,
			\item $\forall (x_a,x_b)\in Q: \quad d(H_a(x_a), H_b(x_b))\leq\epsilon$,
			\item $\forall x_a \in X_a$ with $(x_a, u_a, x_a')\in \underset{a}{\longrightarrow}:\quad (x_a,x_b)\in Q$ $\implies$ $\exists (x_b, u_b, x_b')\in \underset{b}{\longrightarrow}$ such that $(x_a',x_b')\in Q$.
		\end{itemize}
	\end{definition}
	If there exists an $\epsilon$-approximate simulation relation from $\sys_a$ to $\sys_b$, we say that $\sys_b$ $\epsilon$-approximately simulates $\sys_a$ and write $\sys_a \overset{\epsilon}{\preceq}\sys_b$. Finally, we introduce an alternative definition of \textit{power quotient systems} (for the original one, see \cite{tabuada_book}): \begin{comment}If $\epsilon=0$, meaning that $H_a(x_a) = H_b(x_b)$, the relation is called an \textit{exact simulation relation}.\end{comment} 
	\begin{definition}[Power Quotient System \cite{arman_formal_etc}]
		Consider a system $\sys = (X, X_0, U, \longrightarrow, Y, H)$ and an equivalence relation $Q\subseteq X\times X$. The power quotient system of $\sys$ is the tuple $\sys_{/ Q}= (X_{/ Q}, X_{0_{/ Q}}, U_{/ Q}, \underset{{/ Q}}{\longrightarrow}, Y_{/ Q}, H_{/ Q})$, where:
		\begin{itemize}
			\item $X_{/ Q} = X/ Q$,
			\item $X_{0_{/ Q}} = \{x_{/ Q} \in X_{/ Q}:x_{/ Q}\cap X_0 \neq \emptyset \}$,
			\item $U_{/ Q} = U$,
			\item $(x_{/ Q}, u, x'_{/ Q})\in\underset{{/ Q}}{\longrightarrow}$ if $\text{ }\exists (x,u,x')\in \longrightarrow$ such that $x \in x_{/ Q}$ and $x' \in x'_{/ Q}$,
			\item $Y_{/ Q}\subseteq2^Y$,
			\item $H_{/ Q}(x_{/ Q})=\bigcup\limits_{x\in x_{/ Q}}H(x)$.
		\end{itemize}
	\end{definition}
	\begin{lemma}[\hspace{1sp}\cite{arman_formal_etc}] \label{precision lemma}
		Consider a metric system $\sys$, an equivalence relation $Q\subseteq X\times X$ and the power quotient system $\sys_{/ Q}$. For any $\epsilon$ s.t. $\epsilon \geq \sup\limits_{x\in x_{/Q},\text{ }x_{/Q}\in X/Q}d_H(H(x),H_{/ Q}(x_{/ Q}))$, $\sys_{/ Q}$ $\epsilon$-approximately simulates $\sys$, i.e. $\sys \overset{\epsilon}{\preceq}\sys_{/ Q}$
	\end{lemma}
	 	
\subsection{Reachability Analysis}
	\begin{definition}[Reachable Set and Reachable Flowpipe]
		Consider the system \eqref{ode}. Given a set of initial states $\mathcal{I}\subseteq\real^n$, the system's reachable set at time $t_{\star}$ is defined as:% the set of all states that the system can reach at time $t_{\star}$ starting from $\mathcal{I}$ at time $t_0=0$:
		\begin{equation*}
			\reach_{t_{\star}}^{f}(\mathcal{I}) := \{\zeta(t_{\star};x_0):x_0\in \mathcal{I}\}.
		\end{equation*}
		The reachable flowpipe of the system in the time interval $[\underline{\tau},\overline{\tau}]$ is defined as:% the set of all states that the system can reach between $[\underline{\tau},\overline{\tau}]$ starting from $\mathcal{I}$ at $t_0=0$:
		\begin{equation*}
			\reach_{[\underline{\tau},\overline{\tau}]}^{f}(\mathcal{I}) := \bigcup\limits_{t\in[\underline{\tau},\overline{\tau}]} \reach_{t}^{f}(\mathcal{I}).
		\end{equation*}
	\end{definition}	
	Reachability analysis tools (e.g. dReach \cite{dreach}) generally compute overapproximations of reachable sets and flowpipes. They can also check if the computed flowpipe enters an \textit{unsafe set} $\mathcal{U}_f$, i.e. if $\reach_{[\underline{\tau},\overline{\tau}]}^{f}(\mathcal{I}) \cap \mathcal{U}_f \neq \emptyset$. For ease of exposition, we use the same notation for reachable sets/flowpipes and the results of these tools (i.e. their overapproximations).
	
\subsection{Event-Triggered Control Systems}
	Consider the continuous-time control system:
	\begin{equation} \label{ct sys}
		\dot{\zeta}(t) = f(\zeta(t), \upsilon(\zeta(t))),
	\end{equation}
	where $\zeta:\real\to\real^n$, $\upsilon:\real^n\to\real^m$ and $f:\real^n\times\real^m\to\real^n$. The \textit{sample-and-hold} implementation of \eqref{ct sys} is as follows:
	\begin{equation} \label{snh}
		\dot{\zeta}(t) = f(\zeta(t), v(\zeta(t_k))), \quad t\in[t_k,t_{k+1}),
	\end{equation}
	i.e. the input is constant between two consecutive \textit{sampling times} $t_k$, $t_{k+1}$, and is only updated at sampling times. By introducing the \textit{measurement error}:
	\begin{equation*}
		\varepsilon(t) = \zeta(t_k) - \zeta(t), \quad t\in[t_k,t_{k+1}),
	\end{equation*}
	i.e. the deviation of the current state $\zeta(t)$ from the last sampled state $\zeta(t_k)$, we can write \eqref{snh} as:
	\begin{equation}\label{etc system}
		\dot{\zeta}(t) = f(\zeta(t), v(\zeta(t)+\varepsilon(t))), \quad t\in[t_k,t_{k+1}).
	\end{equation}
	In event-triggered control (ETC) the sampling time instants, or \textit{triggering times}, are defined as follows:
	\begin{equation} \label{trig cond}
	t_{k+1} := t_k + \inf\{t>0: \phi(\zeta(t;x_k), \varepsilon(t)) \geq 0 \},
	\end{equation}
	where $x_k$ is the state measurement from the previous sampling time $t_k$ and $\phi(\cdot,\cdot)$ is the \textit{triggering function}. Equation \eqref{trig cond} is the \textit{triggering condition}, and the difference $t_{k+1}-t_{k}$ is called \textit{inter-event time}. Every point $x_k \in\real^n$ in the state space of \eqref{etc system} admits a specific inter-event time $\tau:\real^n\to\real^{+}$:
	\begin{equation}\label{interevent time}
		\tau(x_k) := \inf\{t>0: \phi(\zeta(t;x_k), \varepsilon(t)) \geq 0 \}.
	\end{equation}
	For ease of exposition, we consider the most popular triggering function for nonlinear systems, derived in \cite{tabuada2007etc}:
	\begin{equation}\label{tabuada trig fun}
		\phi(\zeta(t;x_k), \varepsilon(t)) := |\varepsilon(t)|^2 - \sigma^2|\zeta(t;x_{k})|^2,
	\end{equation}
	where $\sigma$ is a constant. According to \cite{tabuada2007etc}, $\phi(\cdot,\cdot)$ is designed such that the ETC implementation \eqref{etc system}-\eqref{trig cond} is globally asymptotically stable, i.e.: $\phi(\zeta(t), \varepsilon(t)) < 0 \implies \dot{V}(\zeta(t)) < 0,$
	where $V(\zeta(t))$ is a Lyapunov function for the ETC system.
	
	We introduce the extended ETC system with state vector $\xi(t) = \begin{bmatrix}\zeta^{\top}(t) &\varepsilon^{\top}(t) \end{bmatrix}^{\top}\in \real^{2n}$ and dynamics:
	\begin{equation} \label{extended etc system}
	\begin{aligned}
	&\dot{\xi}(t) = \begin{bmatrix}
	f(\zeta(t), v(\zeta(t)+\varepsilon(t)))\\ -f(\zeta(t), v(\zeta(t)+\varepsilon(t)))
	\end{bmatrix} = F(\xi(t)), \text{ }  t\in[t_k,t_{k+1}),\\
	&\xi(t_{k+1}) = \begin{bmatrix}
	\zeta(t_{k+1}^-)\\ 0
	\end{bmatrix}.
	\end{aligned}
	\end{equation}
	While $\zeta(t)$ flows continuously for all time, $\varepsilon(t)$ performs jumps at each sampling time, because the state is measured again and the measurement error becomes zero. The reachable sets of the original ETC system \eqref{snh} are the projection of the reachable sets of the extended one \eqref{extended etc system} to the $\zeta$ variables:
	\begin{equation}\label{projection reach}
		\reach_{t_{\star}}^f(\mathcal{I}_f) = \boldsymbol{\pi}_{\zeta}\reach_{t_{\star}}^F(\mathcal{I}_F),
	\end{equation}
	where $\mathcal{I}_F = \{(x,\boldsymbol{0})\in\real^{2n}:\text{ }x\in\mathcal{I}_f\}$.
	
	\subsection{Homogeneous Systems and Scaling of Inter-Event Times}
	
	We recall results derived in \cite{tosample} regarding the scaling law of homogeneous systems' ETC inter-event times. For clarity, we consider the classical notion of homogeneity, with respect to the standard dilation (for more information see \cite{anta2012isochrony}):
	\begin{definition}[Homogeneous Function \cite{anta2012isochrony}]\label{homogeneous function}
		A function $f: \mathbb{R}^n \rightarrow \mathbb{R}^m$ is homogeneous of degree $\alpha \in \real$, if for all $x\in\real^n$ and $\lambda>0$:
		\begin{equation*}
		f(\lambda x) = \lambda^{\alpha+1}f(x).
		\end{equation*}	
	\end{definition}
	A system \eqref{ct sys} is homogeneous of degree $\alpha$, if $f(\zeta,\upsilon(\zeta))=\tilde {f}(\zeta)$ is homogeneous of the same degree. The following theorem dictates how the inter-event times of a homogeneous ETC system scale along its \textit{homogeneous rays}, i.e. lines starting from the origin:
	\begin{theorem}[Scaling Law \cite{tosample}]
		Consider an ETC system \eqref{etc system}-\eqref{trig cond} homogeneous of degree $\alpha$. Let the triggering function  be homogeneous of degree $\theta$. For all $x\in\mathbb{R}^n$, the inter-event times $\tau(\cdot)$ defined by \eqref{interevent time} scale according to:
		\begin{equation} \label{scaling}
		\tau(\lambda x) = \lambda^{-\alpha} \tau(x), \quad \lambda>0,
		\end{equation}
	\end{theorem}
	\begin{assumption} \label{assumption}
		For this work, we assume that:
		\begin{itemize}
			\item The ETC system \eqref{etc system} is homogeneous of degree $\alpha \geq 1$.
			\item The triggering function is the one defined in \eqref{tabuada trig fun}.
		\end{itemize}
	\end{assumption}
	\begin{remark}
		As in \cite{delimpa}, the results of this work are applicable to more general triggering functions $\phi(\cdot)$ satisfying:
		\begin{itemize}
			\item $\phi(\cdot)$ is homogeneous of degree $\theta \geq 1$, with $r_i=1$,
			\begin{comment}
			\item $\phi(\cdot)$ renders the ETC system \eqref{etc system}-\eqref{trig cond} globally asymptotically stable, i.e. $\phi(\zeta(t), \varepsilon(t))<0\implies \dot{V}(\zeta(t))<0$, where $V(\cdot)$ is a Lyapunov function for the system,
			\end{comment}
			\item for all $x\in\mathbb{R}^n\setminus\{\boldsymbol{0}\}$, $\phi(\zeta(0;x), \varepsilon(0))<0$ and $\exists t_x\in(0,+\infty)$ such that $\phi(\zeta(t_x;x), \varepsilon(t_x))=0$.
		\end{itemize}
		In Section VI, the extension to general nonlinear systems and triggering functions is briefly discussed.
	\end{remark}
\section{Problem Statement}
	We aim at constructing traffic abstractions of nonlinear homogeneous ETC systems \eqref{etc system}-\eqref{trig cond}. This task has been already carried out for LTI systems in \cite{arman_formal_etc}. Thus, we adopt a similar problem formulation. We introduce the system 
	\begin{equation}\label{infinite sys}
	\sys = (X, X_0, U, \longrightarrow, Y, H),
	\end{equation}
	where $X,X_0\subseteq\real^n$, $U=\emptyset$ (the system is autonomous), $Y\subseteq\real^+$, $H(x) = \tau(x)$ and the transition relation $\longrightarrow\subseteq X\times X$ is such that $(x,x')\in\longrightarrow \iff \zeta(\tau(x);x)=x'$.
	Observe that the set of output sequences of the above system is the collection of all sequences of inter-event times that the ETC system \eqref{etc system}-\eqref{trig cond} can exhibit, i.e. it captures exactly the traffic generated by the ETC system. However, \eqref{infinite sys} is an infinite-state system and cannot serve as a finite handleable abstraction of the ETC system. This leads us to the following:
	\begin{problem statement}
		Consider the system \eqref{infinite sys}. Construct an equivalence relation $Q\subseteq X\times X$ and a power quotient system $\sys_{/ Q}= (X_{/ Q}, X_{0_{/ Q}}, U_{/ Q}, \underset{{/ Q}}{\longrightarrow}, Y_{/ Q}, H_{/ Q})$ with:
		\begin{itemize}
			\item $X_{/ Q} = X/Q := \{\reg_{1,1},\dots,\reg_{i,j},\dots,\reg_{q,m}\}$,
			\item $X_{0_{/ Q}} = \{\reg_{i,j}: \reg_{i,j}\cap X_0\neq\emptyset \}$,
			\item $U_{/Q}=\emptyset$,
			\item $(x_{/Q},x'_{/Q})\in\underset{{/ Q}}{\longrightarrow}$ if $\exists x \in x_{/Q}$ and $\exists x' \in x'_{/Q}$ such that $\zeta(H(x);x)=x'$,
			\item $Y_{/Q}\subseteq 2^Y = 2^{\real^+}$,
			\item $H_{/ Q}(x_{/Q}) := [\underline{\tau}_{x_{/Q}},\overline{\tau}_{x_{/Q}}]$, with:
					\begin{equation}\label{intervals}
					\underline{\tau}_{x_{/Q}} \leq \inf\limits_{x\in x_{/Q}}H(x), \quad \overline{\tau}_{x_{/Q}} \geq \sup\limits_{x\in x_{/Q}}H(x).
					\end{equation}
		\end{itemize}
	\end{problem statement}
	The reason to use the $i,j$-subscript on $\reg_{i,j}$ will become clear later. Note that: a) the power quotient system's states are regions in the ETC system's state-space, b) a transition in the quotient system takes place when the ETC system triggers and c) the outputs of the quotient system are intervals containing the corresponding outputs of  \eqref{infinite sys}, i.e. the ETC system's inter-event times. Hence, each possible sequence of the ETC system's inter-event times is captured by an output sequence of the power quotient system; the power quotient system abstracts the ETC system's timing behaviour.
	
\section{Constructing the Abstraction}
	In this section, we construct the abstraction $\sys_{/ Q}$, i.e. we construct $Q$, $X_{/ Q}$, $H_{/ Q}(x_{/Q})$ and $\underset{{/ Q}}{\longrightarrow}$. In \cite{arman_formal_etc}, to get $X_{/ Q}$ the state-space is partitioned into a finite number of cones. Afterwards, the bounds $\underline{\tau}_{x_{/Q}},\overline{\tau}_{x_{/Q}}$ for each conic region are computed via LMIs. To obtain the transitions, reachability analysis on each set $x_{/Q}$ is performed..
	
	In this work we adopt a different approach, regarding the partitioning of the state-space and the computation of $\underline{\tau}_{x_{/Q}}$ and $\overline{\tau}_{x_{/Q}}$. In particular, first a finite set of lower bounds on inter-event times $\{\underline{\tau}_{1},\dots,\underline{\tau}_{q}\}$ is fixed, that serve as $\underline{\tau}_{x_{/Q}}$. Then, by considering intersections of cones with \textit{inner approximations of isochronous manifolds}, previously constructed in \cite{delimpa}, the regions $\{\reg_{1,1},\dots,\reg_{i,j},\dots,\reg_{q,m}\}$ are derived such that:
	\begin{equation} \label{lower bounds}
		\forall x \in \reg_{i,j}: \quad \tau(x)\geq \underline{\tau}_{i},
	\end{equation}
	which implies that the first part of \eqref{intervals} is satisfied. Afterwards, to perform reachability analysis on the regions $\reg_{i,j}$, and since they obtain a transcendental representation, we overapproximate them by semi-algebraic ball segments $\hat{\reg}_{i,j}$. Finally, the upper bounds $\overline{\tau}_{x_{/Q}}$ and the transitions $\underset{{/ Q}}{\longrightarrow}$ are determined concurrently, via reachability analysis on $\hat{\reg}_{i,j}$. 
	
	\subsection{Lower Bounds on Inter-Event Times and State-Space Partitioning}
	First, we recall the notion of isochronous manifolds of ETC systems, which was firstly introduced in \cite{anta2012isochrony}:
	\begin{definition}[Isochronous Manifolds]
		Consider an ETC system \eqref{etc system}-\eqref{trig cond}. The set $M_{\tau_{\star}}=\{x\in\mathbb{R}^n : \tau(x)=\tau_{\star}\}$, where $\tau(x)$ is as in \eqref{interevent time}, is called isochronous manifold of time $\tau_{\star}$.
		\label{manifold definition}
	\end{definition}
	In other words, the isochronous manifold $M_{\tau_{\star}}$ consists of all points in the state-space of an ETC system that correspond to the same inter-event time $\tau_{\star}$. Isochronous manifolds are manifolds of dimension $n-1$ (proven in \cite{anta2012isochrony}).
	\begin{proposition}[\hspace{1sp}\cite{anta2012isochrony}]\label{star prop}
		Consider an ETC system \eqref{etc system}-\eqref{trig cond}, and let Assumption \ref{assumption} hold. Each homogeneous ray intersects any isochronous manifold only at one point:
		\begin{equation} \label{star manifolds}
		\forall \tau_{\star}>0 \text{ and }\forall x \in \mathbb{R}^n\setminus\{\boldsymbol{0}\}: \text{ } \exists! \lambda_x>0 \text{ such that } \lambda_x x \in M_{\tau_{\star}}
		\end{equation}
	\end{proposition}
	\begin{proposition}[\hspace{1sp}\cite{delimpa}] \label{encircle prop}
		Consider an ETC system \eqref{etc system}-\eqref{trig cond}, and let Assumption \ref{assumption} hold. Consider isochronous manifolds $M_{\tau_i}$ and $M_{\tau_{i+1}}$, with $\tau_{i}<\tau_{i+1}$. For all $x \in M_{\tau_i}$:
		\begin{equation} \label{encirclement}
		\exists! \lambda_x\in(0,1) \text{ s.t. } \lambda_x x \in M_{\tau_{i+1}} \wedge \not\exists \kappa_x\geq 1 \text{ s.t. } \kappa_x x \in M_{\tau_{i+1}}.
		\end{equation}
	\end{proposition}
	Proposition \eqref{encircle prop} implies that isochronous manifolds that correspond to smaller inter-event times are further away from the origin in every direction. The two above propositions are depicted in Fig. \ref{two manifolds}.
	\begin{figure}[!h]
		\centering
		\includegraphics[width=2in]{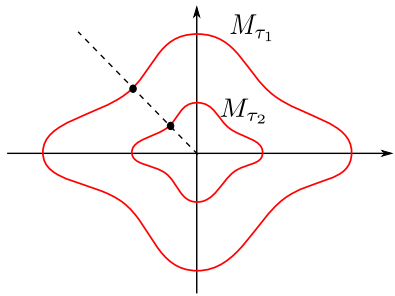}
		\caption{Isochronous manifolds $M_{\tau_1}$, $M_{\tau_2}$ ($\tau_1<\tau_2$). They are intersected by homogeneous rays only once. $M_{\tau_1}$ is further away from the origin in every direction compared to $M_{\tau_2}$.}
		\label{two manifolds}
	\end{figure} 
	Now, consider the region which is enclosed by two isochronous manifolds $M_{\tau_{i}}$ and $M_{\tau_{i+1}}$ with $\tau_i<\tau_{i+1}$. The scaling law \eqref{scaling} directly implies that $\tau_i$ lower bounds the inter-event times of all points in this region, i.e. \eqref{lower bounds} holds. Thus, if we could obtain these regions we would solve the problem of state-space partitioning. However, these exact regions cannot be derived analytically, since nonlinear systems generally do not obtain closed form solutions.
	
	In \cite{delimpa}, inner-approximations $\underline{M}_{\tau_{\star}}$ of isochronous manifolds $M_{\tau_{\star}}$, that satisfy \eqref{star manifolds}, \eqref{encirclement}, were derived analytically. It was shown that the regions enclosed by such approximations do satisfy \eqref{lower bounds}. Hence, we use them to partition the state-space and determine the abstraction's states. Let us recall the method presented in \cite{delimpa}. First, define the sets: 
	\begin{equation*}
	\begin{aligned}
	&\Omega_d:=\{x\in\mathbb{R}^{2n}: |x|<d\} ,\text{ }
	\setz := \{x\in \mathbb{R}^n: V(x) \leq c\} ,\\
	&\sete := \{e\in\mathbb{R}^n: e=x_0-x, \quad x_0,x\in \setz \} ,\text{ }
	\Xi := \setz \times \sete,
	\end{aligned}
	\end{equation*}
	where $d,c>0$ and $V(x)$ is a Lyapunov function for the ETC system \eqref{etc system}. The first step to obtain the inner-approximations is to solve the following feasibility problem:
	\begin{feasibility problem} \label{feasibility problem}
		Find coefficients $\delta_0,\delta_1,\dots,\delta_p\in\real^+_0$ such that:
		\begin{equation*}
			\begin{aligned}	
			&\mathcal{L}_F^p\phi(z,e) \le \sum_{i=0}^{p-1}\delta_i\mathcal{L}_F^i\phi(z,e) + \delta_p, \quad \forall \begin{bmatrix}z^{\top} &e^{\top}\end{bmatrix}\in \Omega_d,\\
			&\delta_0\phi(z,0)+\delta_p \geq \varepsilon > 0, \quad \forall z \in \setz,
			\end{aligned}
			\label{delta feasibility}
		\end{equation*}
		where $p>0$ is a user-defined positive integer, $\varepsilon$ is an arbitrary positive constant, and $c$, $d$ are such that $\Xi\subset \Omega_d$.
	\end{feasibility problem}
	In \cite{delimpa}, a computational algorithm has been developed that solves the above problem. Note that there always exists a solution; e.g. $\delta_p \geq \max\{\epsilon,\sup\limits_{z \in\Omega_d}\mathcal{L}_F^p\phi(z)\}$ and $\delta_i=0$ for $i=0, \dots, p-1$. Having obtained such $\delta_i\in\mathbb{R}_0^+$, the inner-approximations of isochronous manifolds are derived:
	\begin{theorem}[\hspace{1sp}\cite{delimpa}]\label{approx theorem}
		Consider an ETC system \eqref{etc system}-\eqref{trig cond}, a triggering function $\phi(\zeta(t;x), \varepsilon(t))$, and coefficients $\delta_0,\delta_1,\dots,\delta_p$ solving Problem \ref{feasibility problem}. Let Assumption \ref{assumption} hold. Let $D_{\rho}=\{x\in\mathbb{R}^n: |x|=\rho\}$, with $\rho>0$ such that $D_{\rho} \subset \setz$. Define the following function for all $x\in\mathbb{R}^n\setminus\{\boldsymbol{0}\}$:
		\begin{equation} \label{mu}
		\mu(x,t) := C (\tfrac{|x|}{\rho})^{\theta+1} \boldsymbol{e}^{A(\frac{|x|}{\rho})^{\alpha}t} 
		\begin{bmatrix}
		\phi(\rho\tfrac{x}{|x|},0)\\
		\max\bigg( \mathcal{L}_f\phi(\rho\tfrac{x}{|x|},0),0 \bigg)\\
		\vdots\\
		\max\bigg( \mathcal{L}_f^{p-1}\phi(\rho\tfrac{x}{|x|},0),0 \bigg)\\
		\delta_p
		\end{bmatrix},
		\end{equation}
		where:\begin{equation*}
		A = \begin{bmatrix}
		0 &1 &0 &\dots &0 &0\\
		0 &0 &1 &\dots &0 &0\\
		\vdots &\vdots &\qquad &\ddots &\vdots &\vdots\\
		0 &0 &0 &\dots &1 &0\\
		\delta_0 &\delta_1 &\delta_2 &\dots &\delta_{p-1} &1\\
		0 &0 &0 &\dots &0  &0\\
		\end{bmatrix},\text{ }
		C = \begin{bmatrix}
		1 \\0 \\\vdots \\0
		\end{bmatrix}^\top,
		\end{equation*}
		and $\alpha$ and $\theta$ are the degrees of homogeneity of the system and the triggering function, respectively. The set $\underline{M}_{\tau_{\star}}:=\{x\in\mathbb{R}^n:\mu(x,\tau_{\star})=0\}$
		satisfies \eqref{star manifolds} and \eqref{encirclement} and is an inner-approximation of the isochronous manifold $M_{\tau_{\star}}$.
	\end{theorem}
	Regarding the regions enclosed by inner-approximations of isochronous manifolds, we get that \eqref{lower bounds} is satisfied:
	\begin{proposition}[\hspace{1sp}\cite{delimpa}]
		Consider the following set:
		\begin{equation}\label{regions}
			\reg_{i} := \{x\in\mathbb{R}^n: \mu(x,\underline{\tau}_{i+1})>0 , \text{ } \mu(x,\underline{\tau}_{i})\leq0\},
		\end{equation}
		with $0<\underline{\tau}_{i}<\underline{\tau}_{i+1}$. The set $\reg_i$ satisfies \eqref{lower bounds}, i.e. $ \tau(x) \geq \underline{\tau}_i,$ for all $x \in \reg_i$.
	\end{proposition}
	The sets $\reg_i$ are the regions with their outer and inner boundaries being $\underline{M}_{\underline{\tau}_{i}}$ and $\underline{M}_{\underline{\tau}_{{i+1}}}$ respectively (see Fig. \ref{discr approx}).
	\begin{figure}[!h]
		\centering
		\includegraphics[width=2.5in]{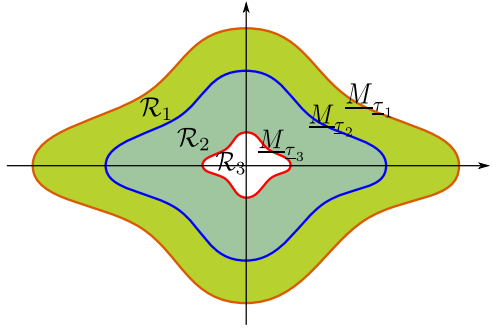}
		\caption{Inner-approximations $\underline{M}_{\underline{\tau}_{i}}$ of isochronous manifolds and the regions enclosed by them $\reg_i$.}
		\label{discr approx}
	\end{figure}
	Since $\reg_i$ satisfy \eqref{lower bounds}, we could directly use them to partition the state-space as required. However, these sets are generally large, which could harm the accuracy of the reachability analysis that is to be conducted afterwards. Thus, we further divide them, using conic intersections. We create a state-space covering by $m$ cones (see \cite{fiter2012state}), which admit the representation:
	\begin{equation}
	\cone_j := \{x\in\real^n: E_jx\succeq 0\}, \quad j=1,\dots,m.
	\end{equation}
	The sets $\{\reg_{1,1},\dots,\reg_{q,m}\}$ are obtained as intersections of regions $\reg_i$ with cones $\mathcal{C}_j$: We fix a set of $q$ times $\{\underline{\tau}_{1},\dots,\underline{\tau}_{q}\}$ that serve as lower bounds on inter-event times, and obtain the regions $\{\reg_1,\dots,\reg_q\}$ from \eqref{regions}. Then, employing a covering by $m$ cones, we derive the sets $\reg_{i,j}$ as
	\begin{equation}
		\reg_{i,j} := \reg_{i}\cap\mathcal{C}_{j},
	\end{equation}
	where $i=1,\dots,q$ and $j=1,\dots,m$ (see Fig. \ref{discr cones}). Since $\reg_{i,j}\subset\reg_{i}$, from \eqref{lower bounds} we get that $\tau(x)\geq\underline{\tau}_i$ for all $x\in\reg_{i,j}$. Thus, we can fix:
	\begin{equation}
		\underline{\tau}_{\reg_{i,j}} := \underline{\tau}_{i}.
	\end{equation}
	\begin{figure}[!h]
		\centering
		\includegraphics[width=2.5in]{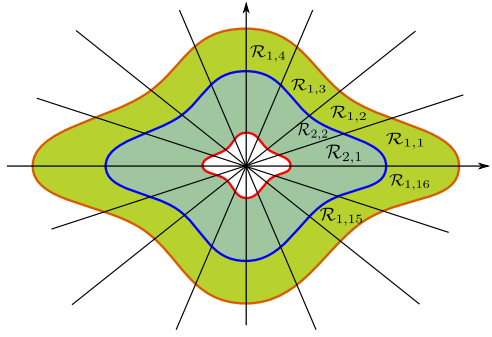}
		\caption{State-space partitioning into regions $\reg_{i,j}$.}
		\label{discr cones}
	\end{figure}
	It is straightforward to design the equivalence relation $Q$ as: $Q := \{(x, x_{/Q})\in \real^n\times\real^n:x\in x_{/Q}, \text{ } x_{/Q} \in \{\reg_{1,1},\dots,\reg_{q,m}\}\}$.
	\begin{remark}
		To define the regions $\reg_{q,j}$ (e.g. see the white innermost regions in Fig. \ref{discr cones}), we cannot use \eqref{regions}, since there is no $\underline{\tau}_{q+1}$. Thus, we define $\reg_{q} := \{x\in\mathbb{R}^n: \mu(x,\underline{\tau}_{q})\leq0\}$ and $\reg_{q,j}:=\reg_{q}\cap\cone_j$. Observe that \eqref{lower bounds} still holds.
	\end{remark}
	
	\subsection{Overapproximations of the sets $\reg_{i,j}$} \label{overapprox section}
	 To obtain the upper bounds $\overline{\tau}_{x_{/Q}}$ and the state transitions, reachability analysis on the regions $\reg_{i,j}$ is conducted. However, it is obvious from \eqref{mu} and \eqref{regions} that the sets $\reg_{i,j}$ are transcendental, which renders their computational handling very difficult. To the authors' knowledge, there are no reachability analysis tools that can handle effectively such sets. Hence, we have to overapproximate them. 
	 
	 In general, the overapproximation of transcendental sets is very challenging. However, leveraging special characteristics of the specific representation we devised an algorithm that overapproximates the sets $\reg_{i,j}$ by ball segments (Fig. \ref{ball segment}): 
	 \begin{equation}\label{ball seg eq}
	 \hat{\reg}_{i,j} :=\{x\in\mathcal{C}_{j}: \underline{r}_{i+1,j}\leq|x|\leq\overline{r}_{i,j}\}.
	 \end{equation}
	 Note that, since  $\reg_{i,j}\subseteq\hat{\reg}_{i,j}$, the following holds:
	 \begin{equation}\label{reg approx reach}
	 \reach^f_{[\underline{\tau}_{\reg_{i,j}},\overline{\tau}_{\reg_{i,j}}]}(\reg_{i,j})\subseteq\reach^f_{[\underline{\tau}_{\reg_{i,j}},\overline{\tau}_{\reg_{i,j}}]}(\hat{\reg}_{i,j}).
	 \end{equation}
	 \begin{figure}[!h]
	 	\centering
	 	\includegraphics[width=2in]{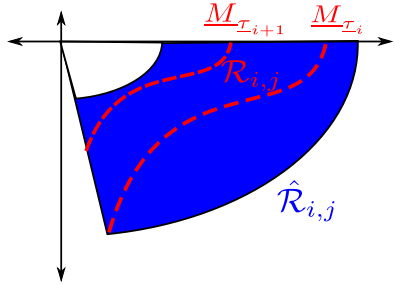}
	 	\caption{Ball segment $\hat{\reg}_{i,j}$ (blue region) overapproximating $\reg_{i,j}$ (region delimited by the red lines and the cone).}
	 	\label{ball segment}
	 \end{figure}
	 To obtain the ball segments \eqref{ball seg eq}, $\underline{r}_{i+1,j}$ and $\overline{r}_{i,j}$ must be determined; i.e. spherical segments (intersections of spheres with cones) that inner- and outer- approximate the conic sections $M_{\underline{\tau}_{i+1}}\cap\cone_j$ and $M_{\underline{\tau}_i}\cap\cone_j$, respectively, have to be found (see Fig. \ref{ball segment}). A whole sphere $\sphere_{r} :=\{x\in\real^n: |x|=r\}$ inner-approximates the whole $\underline{M}_{\underline{\tau}_{\star}}$ if it lies entirely in the region enclosed by $\underline{M}_{\underline{\tau}_{\star}}$, that is if $\mu(x,\underline{\tau}_{\star})\leq0$ for all $x\in\sphere_r$. 
	 Likewise, a spherical segment $\sphere_{\underline{r}_{\star,j}}\cap\cone_j$ inner-approximates $\underline{M}_{\underline{\tau}_{\star}}\cap\cone_j$ if the following holds: 
	 \begin{equation} \label{inner sphere}
		 \forall x\in \sphere_{\underline{r}_{\star,j}}\cap\cone_j: \quad \mu(x,\underline{\tau}_{\star})\leq0.
	 \end{equation}
	 
	 Formulas like \eqref{inner sphere} can be verified or disproved by SMT solvers, e.g. dReal \cite{dreal}. Thus, a bisection algorithm on $\underline{r}_{\star,j}$ could be employed, by iteratively checking \eqref{inner sphere}. In our case though, $\mu(x,\underline{\tau}_{\star})\leq0$ implies the numerically non-robust symbolic computation of the matrix exponential $\boldsymbol{e}^{A(\frac{|x|}{\rho})^{\alpha}\underline{\tau}_{\star}}$ over the symbolic variable $x$. Luckily, since we want to verify $\mu(x,\underline{\tau}_{\star})\leq0$ on a spherical segment $\sphere_{\underline{r}_{\star,j}}\cap\cone_j$, we can fix $|x|\gets \underline{r}_{\star,j}$, which renders the symbolic matrix exponential a regular numerical one and severely relaxes computations, i.e. $\boldsymbol{e}^{A(\frac{|x|}{\rho})^{\alpha}\underline{\tau}_{\star}} = \boldsymbol{e}^{A(\frac{\underline{r}_{\star,j}}{\rho})^{\alpha}\underline{\tau}_{\star}}$ for all $x \in \sphere_{\underline{r}_{\star,j}}\cap\cone_j$. 
	 This is done by fixing the first argument of $\mu(\cdot,\underline{\tau}_{\star})$ as: $\mu(\frac{x}{|x|}\underline{r}_{\star,j},\underline{\tau}_{\star})$. Consequently, in order to find a spherical inner-approximation $\sphere_{\underline{r}_{\star,j}}\cap\cone_j$ of the conic section $\underline{M}_{\underline{\tau}_{\star}}\cap\cone_j$, we employ a bisection algorithm on the radius $\underline{r}_{\star,j}$ and check iteratively by an SMT solver the following condition: 
	 \begin{equation}\label{smt check inner approx}
	 \forall x \in \sphere_{\underline{r}_{\star,j}}\cap\cone_j: \quad \mu(\frac{x}{|x|}\underline{r}_{\star,j},\underline{\tau}_{\star}) \leq 0.
	 \end{equation}
	 By reversing inequality \eqref{smt check inner approx} we determine an outer-approximation $\sphere_{\overline{r}_{\star,j}}\cap\cone_j$ of $\underline{M}_{\underline{\tau}_{\star}}\cap\cone_j$. Fig. \ref{balls approx man} shows spherical inner/outer-approximations of $\underline{M}_{\underline{\tau}_{\star}}$ for each conic section.
	 \begin{figure}[h]
	 	\centering
	 	\includegraphics[width=1.8in]{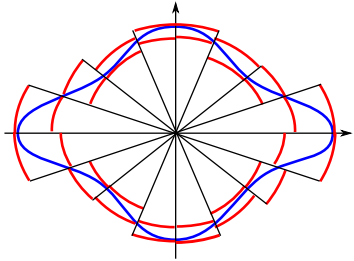}
	 	\caption{Spherical inner- and outer-approximations for each conic section of $\underline{M}_{\underline{\tau}_{\star}}$.}
	 	\label{balls approx man}
	 \end{figure}
	 
	 Until now, we have only done this for one $\underline{M}_{\underline{\tau}_{\star}}$. To derive such approximations for all $\underline{M}_{\underline{\tau}_{i}}$, first observe from \eqref{mu} that $\mu(\lambda x, t) = \lambda^{\theta+1}\mu(x, \lambda^{\alpha}t)$. This implies that
	 $x \in \underline{M}_{\underline{\tau}_{\star}} \implies \lambda x \in \underline{M}_{\lambda^{-\alpha}\underline{\tau}_{\star}}$,
	 i.e. in every direction of the state-space the sets $\underline{M}_{\underline{\tau}_{\star}}$ scale according to their corresponding time $\underline{\tau}_{\star}$. Thus, if $\sphere_{\underline{r}_{\star,j}}\cap\cone_j$ is an inner-approximation of $\underline{M}_{\underline{\tau}_{\star}}\cap\cone_j$, then $\sphere_{\lambda \underline{r}_{\star}}\cap\cone_j$ is an inner-approximation of $\underline{M}_{\lambda^{-\alpha}\underline{\tau}_{\star}}\cap\cone_j$. Consequently, to obtain inner- and outer-approximations of all conic sections $\underline{M}_{\underline{\tau}_{i}}\cap\cone_j$ ($i=1,\dots,q$), we scale the obtained radii accordingly by corresponding factors $\lambda_i = \Big(\frac{\underline{\tau}_{\star}}{\underline{\tau}_{i}}\Big)^{\frac{1}{\alpha}}$,
	 so that we get:
	 $\underline{r}_{i,j}=\lambda_i\underline{r}_{\star,j}$ and $\overline{r}_{i,j}=\lambda_i\overline{r}_{\star,j}$.
	 Finally, as soon as all radii $\underline{r}_{i,j}$, $\overline{r}_{i,j}$ are obtained, the regions $\reg_{i,j}$ are overapproximated by ball segments $\hat{\reg}_{i,j}$ \eqref{ball seg eq}.
	
	\begin{remark}
		For sets $\reg_{q,j}$, which contain the origin, there is no radius $\underline{r}_{q+1,j}$ that defines $\hat{\reg}_{q,j}$ as in \eqref{ball seg eq}, since there is no $\underline{M}_{\underline{\tau}_{q+1}}$. For these sets, we define the overapproximations as:
			$\hat{\reg}_{q,j} = \{x\in\mathcal{C}_{j}: |x|\leq\overline{r}_{q,j}\}.$
	\end{remark}
	\subsection{Upper Bounds on Inter-Event Times and State Transitions}
	
	To complete the construction, what remains is to obtain the upper bounds $\overline{\tau}_{\reg_{i,j}}$ and the state transitions. Let us recall the definition of transitions for the power quotient system: 
	\begin{align*}
	(x_{/Q},x'_{/Q})&\in\underset{{/ Q}}{\longrightarrow} \text{ if:}\\ \exists x \in x_{/Q} \text{ and } \exists x' \in x'_{/Q}&\text{ such that }\zeta(H(x);x)=x'.
	\end{align*}
	To determine such transitions, the inter-event time $\tau(x)=H(x)$ has to be known a priori and the trajectories starting from all $x\in x_{/Q}$ have to be computed,  which is impossible. Thus, we choose to relax the definition of transitions as:
	\begin{align*}
	(x_{/Q},x'_{/Q})&\in\underset{{/ Q}}{\longrightarrow} \text{ if:}\\ \exists x \in x_{/Q} \text{ and } \exists x' \in x'_{/Q}&\text{ such that }x' \in \reach^f_{H_{/Q}(x_{/Q})}(x),
	\end{align*}
	i.e. if there exists a state $x'$ in the region $x'_{/Q}$ that is reachable from at least one state $x \in x_{/Q}$ in the interval $H_{/ Q}(x_{/Q}) := [\underline{\tau}_{x_{/Q}},\overline{\tau}_{x_{/Q}}]$. This characterization of transitions involves conducting reachability analysis and computing set-intersections on the transcendental sets $\reg_{i,j}$. Thus, we relax the characterization once more, motivated by \eqref{reg approx reach}, using the overapproximations $\hat{\reg}_{i,j}$:
	\begin{equation}\label{transitions}
	(\reg_{i,j},\reg_{a,b})\in\underset{{/ Q}}{\longrightarrow} \text{ if:}\quad \reach^f_{[\underline{\tau}_{\reg_{i,j}},\overline{\tau}_{\reg_{i,j}}]}(\hat{\reg}_{i,j}) \cap \hat{\reg}_{a,b} \neq \emptyset .
	\end{equation}
	Hence, if the upper bounds $\overline{\tau}_{\reg_{i,j}}$ are known, we can determine transitions from each $\reg_{i,j}$ by the above formula. In this subsection, we show how to determine the upper bounds and the transitions \textit{concurrently}, via reachability analysis.
	 
	For a certain region $\hat{\reg}_{i,j}$ we fix the reachability analysis time interval as $[\underline{\tau}_{\reg_{i,j}},\tau_{\max}]$, where $\tau_{\max}$ is user-specified. Reachability analysis is conducted on the extended state ETC system \eqref{extended etc system}. The following two sets are defined:
	\begin{align}
	&\mathcal{I}_F := \{(x,\boldsymbol{0})\in\real^{2n}:\text{ }x\in\hat{\reg}_{i,j}\},\label{initial set}\\
	&\mathcal{U}_F := \{(x,e)\in\real^{2n}: \phi(x,e)\leq 0\},\label{unsafe set}
	\end{align}
	where $\phi$ is the triggering function. The set $\mathcal{I}_F$ serves as the initial set, while $\mathcal{U}_F$ serves as the unsafe set. By checking if the system's trajectories are in $\mathcal{U}_F$ at time $\tau_{\max}$, we check if $\tau_{\max}$ is an upper bound on inter-event times:
	\begin{proposition}
		If $\reach_{\tau_{\max}}^F(\mathcal{I}_F)\cap\mathcal{U}_F=\emptyset$, where $F(\cdot)$ is defined in \eqref{extended etc system} and $\mathcal{I}_F$ and $\mathcal{U}_F$ are defined in \eqref{initial set}-\eqref{unsafe set}, then:
		\begin{equation} \label{upper bounds}
			\forall x \in \reg_{i,j}: \quad \tau(x)\leq\tau_{\max}.
		\end{equation}
	\end{proposition} 
	\begin{proof}
		If $\reach_{\tau_{\max}}^F(\mathcal{I}_F)\cap\mathcal{U}_F=\emptyset$, then for all $x\in\hat{\reg}_{i,j}$: $\phi(\zeta(\tau_{\max};x),\varepsilon(\tau_{\max})) > 0$. Thus, $\tau(x)\leq\tau_{\max}$ for all $x\in\hat{\reg}_{i,j}$, and since $\reg_{i,j}\subseteq\hat{\reg}_{i,j}$, we get \eqref{upper bounds}.
	\end{proof} Obviously, in order to find an upper-bound on inter-event times for a region $\reg_{i,j}$, a line search on $\tau_{max}$ is needed. As soon as the upper-bounds $\overline{\tau}_{\reg_{i,j}}$ are obtained, to determine the transitions from each $\reg_{i,j}$ we employ equations \eqref{projection reach}, \eqref{transitions} and the computed flowpipes $\reach^F_{[\underline{\tau}_{\reg_{i,j}}, \overline{\tau}_{\reg_{i,j}}]}(\mathcal{I}_F)$:
	\begin{equation}
	 \reach^f_{[\underline{\tau}_{\reg_{i,j}},\overline{\tau}_{\reg_{i,j}}]}(\hat{\reg}_{i,j}) \cap \hat{\reg}_{a,b} = \boldsymbol{\pi}_{\zeta}\reach^F_{[\underline{\tau}_{\reg_{i,j}}, \overline{\tau}_{\reg_{i,j}}]}(\mathcal{I}_F)\cap \hat{\reg}_{a,b}.
	\end{equation}
	
	\begin{remark}
		Instead of computing timing upper bounds for all $\reg_{i,j}$, one could compute them only for regions $\reg_{i,j}$ for a fixed $i$ and for all $j$, and then use the scaling law \eqref{scaling} to determine the upper bounds for all $\reg_{i,j}$. Also, flows of homogeneous systems scale as: $\zeta(t;\lambda x) = \lambda\zeta(\lambda^{\alpha}t;x)$, which could similarly be employed to determine transitions for all $\reg_{i,j}$, based on transitions of regions $\reg_{i,j}$ for a fixed $i$.
	\end{remark}
	\begin{remark}\label{upper bounds region 1}
		For regions $\reg_{q,j}$, which contain the origin, there is no upper bound on inter-event times, since the origin's inter-event time is theoretically $\infty$. For these regions, we arbitrarily dictate an upper bound $\overline{\tau}_{\reg_{q,j}}\geq\underline{\tau}_{\reg_{q,j}}$ and force the sensors to close the sampling loop whenever the system's last measured state $x_{k-1}\in\reg_{q,j}$ and $t=t_{k-1}+\overline{\tau}_{\reg_{q,j}}$.
	\end{remark}

	Finally, for the constructed abstraction we have:
	\begin{proposition} \label{precision of abstraction}
		The constructed metric system $\sys_{/ Q}$ $\epsilon$-approximately simulates \eqref{infinite sys}, with $\epsilon \leq \max\limits_{i,j}\{\overline{\tau}_{\reg_{i,j}}-\underline{\tau}_{\reg_{i,j}}\}$.
	\end{proposition}
	\begin{proof}
			It is a direct result of Lemma \ref{precision lemma}.
	\end{proof}
	Thus, the constructed abstraction can be used for scheduling, as described in \cite{dieky2015symbolic}.
\section{Numerical Example}
	Consider the homogeneous of degree $2$ nonlinear system:
	\begin{equation}
		\dot{\zeta}_1 = \zeta_1^3+\zeta_1\zeta_2^2, \quad
		\dot{\zeta}_2 = \zeta_1\zeta_2^2 -\zeta_1^2\zeta_2 + \upsilon
	\label{example1}
	\end{equation}
	with $\upsilon(\zeta)=-\zeta_2^3-\zeta_1\zeta_2^2$. A triggering function, used in \cite{anta2012isochrony}, rendering the ETC implementation asymptotically stable is: 
	\begin{equation*}
		\phi(\zeta(t;x), \varepsilon(t)) = |\varepsilon(t)|^2-0.0127^2\cdot 0.3^2|\zeta(t;x)|^2,
	\end{equation*}	
	
	For the abstraction, we fix $\{\underline{\tau}_1, \underline{\tau}_2, \underline{\tau}_3\} = \{4, 8, 20\}\cdot10^{-4}$, which serve as timing lower bounds $\underline{\tau}_{\reg_{i,j}}$ of the abstraction's regions, and the number of cones $m=16$. The abstraction is composed of 48 regions $\reg_{i,j}$. Fig. \ref{example regions} depicts the state-space partitioning into regions $\reg_{i,j}$ created by the cones $\cone_j$ (black rays) and the approximations of isochronous manifolds $\underline{M}_{\underline{\tau}_i}$ (blue curves), which were derived using Theorem \ref{approx theorem} and the computational algorithm of \cite{delimpa}. The red spherical segments and the corresponding cones, show the boundaries of the overapproximations $\hat{\reg}_{2,j}$ for regions $\reg_{2,j}$, obtained as described in Section \ref{overapprox section}. The sets $\hat{\reg}_{i,j}$ are relatively accurate overapproximations of $\reg_{i,j}$. The accuracy can be improved by increasing the number of cones or reducing the bisection's step size, at the expense of heavier computations.
	\begin{figure}[!h]
		\centering
		\includegraphics[width=3in]{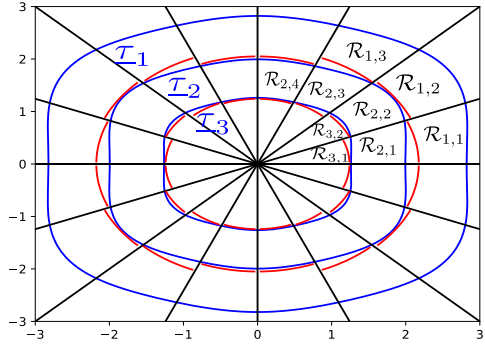}
		\caption{State-space partitioning into 48 regions $\reg_{i,j}$, delimited by 16 cones $\cone_j$ (black rays) and 3 approximations of isochronous manifolds $\underline{M}_{\underline{\tau}_i}$ (blue curves), for times $\{\underline{\tau}_1, \underline{\tau}_2, \underline{\tau}_3\} = \{4, 8, 20\}\cdot10^{-4}$. The red segments delimit the overapproximations $\hat{\reg}_{2,j}$ of regions $\reg_{2,j}$.}
		\label{example regions}
	\end{figure}
	
	Fig. \ref{times_example} shows the timing lower bounds $\underline{\tau}_{\reg_{i,j}}$, which were predefined, and upper bounds $\overline{\tau}_{\reg_{i,j}}$ for each region $\reg_{i,j}$, which were obtained via reachability analysis, as described in Section 4.3. Recall from Remark \ref{upper bounds region 1} that the timing upper bounds for the regions $\reg_{3,j}$ are fixed arbitrarily, such that $\overline{\tau}_{\reg_{3,j}}\geq\underline{\tau}_{\reg_{3,j}}$. Here, we fixed them in such a way that they follow the spatial trend of the upper bounds $\overline{\tau}_{\reg_{i,j}}$ ($i\neq3$). By Proposition \ref{precision of abstraction}, the abstraction's precision is $\epsilon \leq 0.0035$. 
	\begin{figure}[!h]
		\centering
		\includegraphics[width=3in]{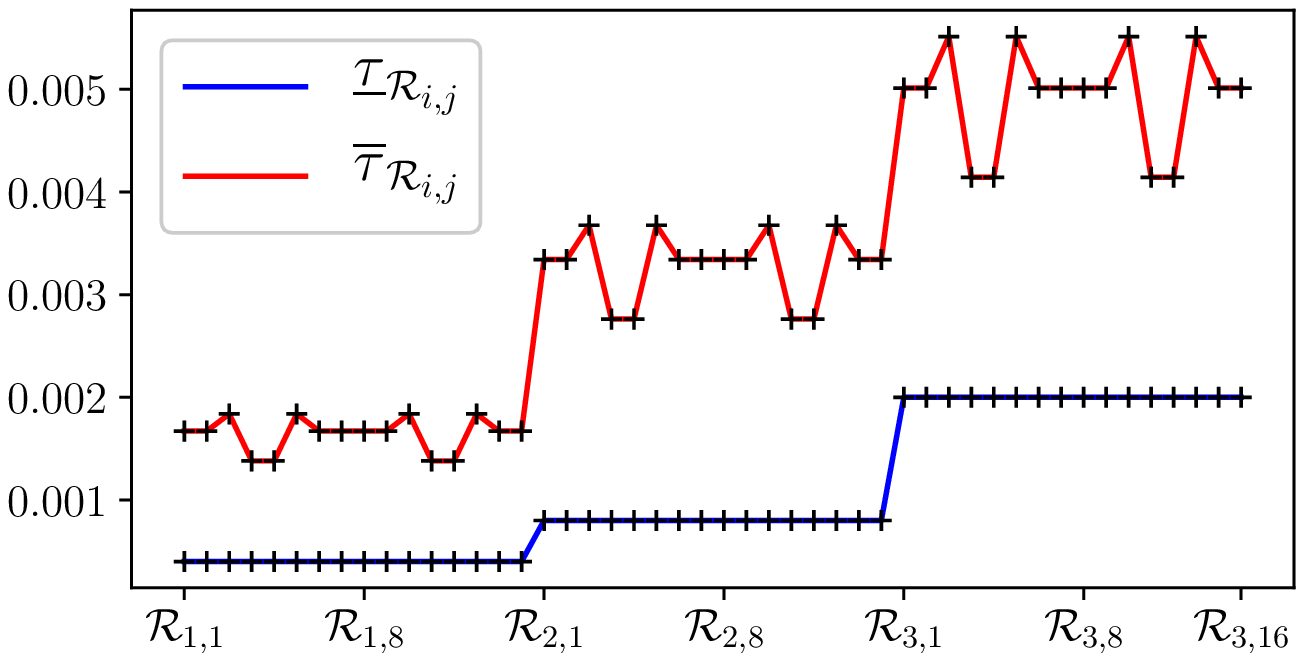}
		\caption{Lower bounds $\underline{\tau}_{\reg_{i,j}}$ and upper bounds $\overline{\tau}_{\reg_{i,j}}$ of inter-event times for each region $\reg_{i,j}$.}
		\label{times_example}
	\end{figure}
	
	In Fig. \ref{transitions_example}, each dotted point $(a,b)$ denotes a transition from region $a$ to region $b$. First, we observe that there exists a transition from each region $\reg_{3,j}$ to any region $\reg_{3,k}$. This is expected, as these regions intersect at the origin. However, note that each one generally corresponds to a different set of transitions. Hence, they do serve as distinct states of the abstraction. Overall, there are 536 transitions. The reachability analysis was carried out with dReach \cite{dreach}.
	\begin{figure}[!h]
		\centering
		\includegraphics[width=3in]{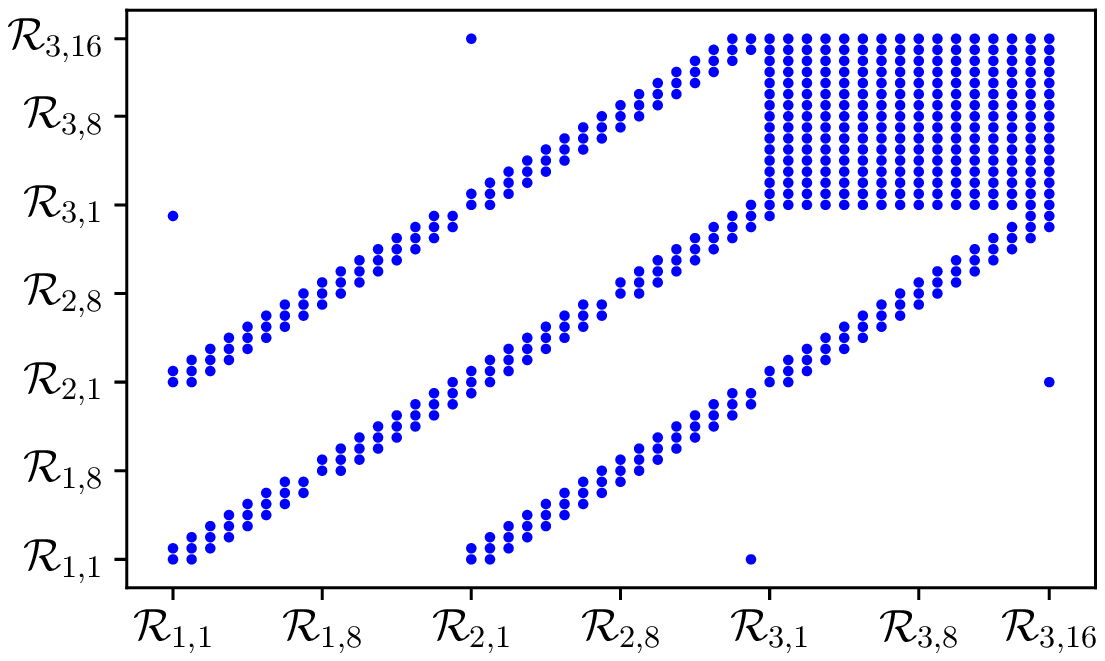}
		\caption{Transitions between regions $\reg_{i,j}$.}
		\label{transitions_example}
	\end{figure}  
	
	Finally, we carried out a simulation to verify our results. The initial condition is set to $\begin{bmatrix}
	1.5 &2
	\end{bmatrix}^{\top}$ and the simulation duration is $0.8$s. The red line in Fig. \ref{sim_times} shows the evolution of inter-event times of the ETC system, and the blue lines represent the bounding intervals $[\underline{\tau}_{\reg_{i,j}}, \overline{\tau}_{\reg_{i,j}}]$ generated by the abstraction, i.e. its output sequence. It is obvious that the abstraction's output sequence confines the ETC system's inter-event times, as expected. Moreover, the system's trajectory starting from $\reg_{1,3}$ followed the spatial path: $\reg_{1,3}\to\reg_{1,3}\to\dots\to\reg_{1,2}\to\reg_{1,2}\to\dots\to\reg_{2,2}\to\reg_{2,2}\to\dots\to\reg_{3,2}\to\reg_{3,2}\to\dots$. Indeed, Fig. \ref{transitions_example} shows that the followed path is contained in the abstraction's transition set.
	\begin{figure}[!h]
		\centering
		\includegraphics[width=3in]{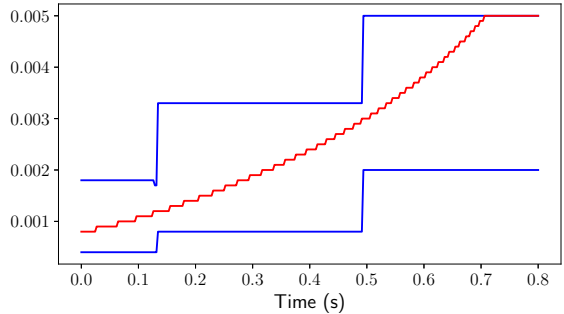}
		\caption{The  ETC system's inter-event times' evolution (red line) and the bounding intervals $[\underline{\tau}_{\reg_{i,j}}, \overline{\tau}_{\reg_{i,j}}]$ generated by the abstraction (blue lines) during a simulation.}
		\label{sim_times}
	\end{figure}

\section{Discussion and Future Work}\label{discussion section}
	We constructed abstractions of nonlinear homogeneous ETC systems that can be employed for traffic scheduling in NCS, as shown in \cite{dieky2015symbolic}, thus contributing to the solution of a prominent problem of ETC. Next step is extending this method to general nonlinear systems and triggering functions. For this purpose, the procedure proposed in \cite{anta2012isochrony} can be used, which renders any system/triggering function homogeneous by embedding it into $\real^{n+1}$ and adding an extra variable $w$. In this case, the original system's trajectories are the ones of the extended homogeneous one confined to the $w=1$ plane. Thus, approximations of the extended system's isochronous manifolds could be used (see \cite{anta2012isochrony,delimpa}). However, new challenges arise, as e.g. the extended system's isochronous manifolds obtain a singularity at the origin.
	\bibliography{bibliography_draft.bib}

\begin{thebibliography}{10}
\providecommand{\url}[1]{#1}
\csname url@samestyle\endcsname
\providecommand{\newblock}{\relax}
\providecommand{\bibinfo}[2]{#2}
\providecommand{\BIBentrySTDinterwordspacing}{\spaceskip=0pt\relax}
\providecommand{\BIBentryALTinterwordstretchfactor}{4}
\providecommand{\BIBentryALTinterwordspacing}{\spaceskip=\fontdimen2\font plus
\BIBentryALTinterwordstretchfactor\fontdimen3\font minus
  \fontdimen4\font\relax}
\providecommand{\BIBforeignlanguage}[2]{{%
\expandafter\ifx\csname l@#1\endcsname\relax
\typeout{** WARNING: IEEEtran.bst: No hyphenation pattern has been}%
\typeout{** loaded for the language `#1'. Using the pattern for}%
\typeout{** the default language instead.}%
\else
\language=\csname l@#1\endcsname
\fi
#2}}
\providecommand{\BIBdecl}{\relax}
\BIBdecl

\bibitem{arzen1999}
K.-E. {\AA}arz{\'e}n, ``A simple event-based pid controller,'' \emph{IFAC
  Proceedings Volumes}, vol.~32, no.~2, pp. 8687--8692, 1999.

\bibitem{tabuada2007etc}
P.~Tabuada, ``Event-triggered real-time scheduling of stabilizing control
  tasks,'' \emph{IEEE Transactions on Automatic Control}, vol.~52, no.~9, pp.
  1680--1685, 2007.

\bibitem{girard2015dynamicetc}
A.~Girard, ``Dynamic triggering mechanisms for event-triggered control,''
  \emph{IEEE Transactions on Automatic Control}, vol.~60, no.~7, pp.
  1992--1997, 2015.

\bibitem{velasco2003self}
M.~Velasco, J.~Fuertes, and P.~Marti, ``The self triggered task model for
  real-time control systems,'' in \emph{Work-in-Progress Session of the 24th
  IEEE Real-Time Systems Symposium}, vol. 384, 2003.

\bibitem{tosample}
A.~Anta and P.~Tabuada, ``To sample or not to sample: Self-triggered control
  for nonlinear systems,'' \emph{IEEE Transactions on Automatic Control},
  vol.~55, no.~9, pp. 2030--2042, 2010.

\bibitem{fiter2012state}
C.~Fiter, L.~Hetel, W.~Perruquetti, and J.-P. Richard, ``A state dependent
  sampling for linear state feedback,'' \emph{Automatica}, vol.~48, no.~8, pp.
  1860--1867, 2012.

\bibitem{delimpa}
G.~Delimpaltadakis and M.~Mazo~Jr, ``Isochronous partitions for region-based
  self-triggered control,'' \emph{arXiv preprint arXiv:1904.08788}, 2019.

\bibitem{2012introtoetc_stc}
W.~P. M.~H. Heemels, K.~H. Johansson, and P.~Tabuada, ``An introduction to
  event-triggered and self-triggered control,'' in \emph{Proceedings of the
  IEEE Conference on Decision and Control}, 2012, pp. 3270--3285.

\bibitem{buttazzo1998elastic}
G.~C. Buttazzo, G.~Lipari, and L.~Abeni, ``Elastic task model for adaptive rate
  control,'' in \emph{Proceedings 19th IEEE Real-Time Systems Symposium (Cat.
  No. 98CB36279)}.\hskip 1em plus 0.5em minus 0.4em\relax IEEE, 1998, pp.
  286--295.

\bibitem{caccamo2000elastic}
M.~Caccamo, G.~Buttazzo, and L.~Sha, ``Elastic feedback control,'' in
  \emph{Proceedings 12th Euromicro Conference on Real-Time Systems. Euromicro
  RTS 2000}.\hskip 1em plus 0.5em minus 0.4em\relax IEEE, 2000, pp. 121--128.

\bibitem{bhattacharya2004anytime}
R.~Bhattacharya and G.~J. Balas, ``Anytime control algorithm: Model reduction
  approach,'' \emph{Journal of Guidance, Control, and Dynamics}, vol.~27,
  no.~5, pp. 767--776, 2004.

\bibitem{fontanelli2008scheduling}
D.~Fontanelli, L.~Greco, and A.~Bicchi, ``Anytime control algorithms for
  embedded real-time systems,'' \emph{Lecture Notes in Computer Science
  (including subseries Lecture Notes in Artificial Intelligence and Lecture
  Notes in Bioinformatics)}, vol. 4981 LNCS, pp. 158--171, 2008.

\bibitem{areqi2015scheduling}
S.~Al-Areqi, D.~Görges, and S.~Liu, ``Event-based networked control and
  scheduling codesign with guaranteed performance,'' \emph{Automatica},
  vol.~57, pp. 128--134, 2015.

\bibitem{lu2002feedback}
C.~Lu, J.~A. Stankovic, S.~H. Son, and G.~Tao, ``Feedback control real-time
  scheduling: Framework, modeling, and algorithms,'' \emph{Real-Time Systems},
  vol.~23, no. 1-2, pp. 85--126, 2002.

\bibitem{cervin2005control}
A.~Cervin and J.~Eker, ``Control-scheduling codesign of real-time systems: The
  control server approach,'' \emph{Journal of Embedded Computing}, vol.~1,
  no.~2, pp. 209--224, 2005.

\bibitem{arman_formal_etc}
A.~S. Kolarijani and M.~Mazo, ``Formal traffic characterization of lti
  event-triggered control systems,'' \emph{IEEE Transactions on Control of
  Network Systems}, vol.~5, no.~1, pp. 274--283, 2016.

\bibitem{tabuada_book}
P.~Tabuada, \emph{Verification and control of hybrid systems: a symbolic
  approach}.\hskip 1em plus 0.5em minus 0.4em\relax Springer Science \&
  Business Media, 2009.

\bibitem{dieky2015symbolic}
A.~S. Kolarijani, D.~Adzkiya, and M.~Mazo, ``Symbolic abstractions for the
  scheduling of event-triggered control systems,'' in \emph{2015 54th IEEE
  Conference on Decision and Control (CDC)}.\hskip 1em plus 0.5em minus
  0.4em\relax IEEE, 2015, pp. 6153--6158.

\bibitem{dreach}
S.~Kong, S.~Gao, W.~Chen, and E.~Clarke, ``dreach: $\delta$-reachability
  analysis for hybrid systems,'' in \emph{International Conference on TOOLS and
  Algorithms for the Construction and Analysis of Systems}.\hskip 1em plus
  0.5em minus 0.4em\relax Springer, 2015, pp. 200--205.

\bibitem{combinatorial}
G.~Ewald, \emph{Combinatorial convexity and algebraic geometry}.\hskip 1em plus
  0.5em minus 0.4em\relax Springer Science \& Business Media, 2012, vol. 168.

\bibitem{anta2012isochrony}
{A. Anta and P. Tabuada}, ``Exploiting isochrony in self-triggered control,''
  \emph{IEEE Transactions on Automatic Control}, vol.~57, no.~4, pp. 950--962,
  2012.

\bibitem{dreal}
S.~Gao, S.~Kong, and E.~M. Clarke, ``dreal: An smt solver for nonlinear
  theories over the reals,'' in \emph{International Conference on Automated
  Deduction}.\hskip 1em plus 0.5em minus 0.4em\relax Springer, 2013, pp.
  208--214.

\end{thebibliography}
	\bibliographystyle{IEEEtran}
\end{document}